\documentclass[conference]{IEEEtran}
\IEEEoverridecommandlockouts
\usepackage{graphicx}
\usepackage{amssymb}
\usepackage{amsmath}
\usepackage{cite}
\usepackage{subfigure}
\usepackage{mathrsfs}
\usepackage[displaymath,mathlines]{lineno}
\usepackage{color}
\usepackage{tabulary}
\usepackage{multirow}
\usepackage{algpseudocode}
\usepackage{algorithm}
\usepackage{algorithmicx}
\usepackage{pbox}
\usepackage{multicol}
\usepackage{lipsum}
\usepackage{amsthm}
\usepackage{relsize}
\usepackage{lipsum}
\usepackage{epstopdf}
\usepackage{tabularx,ragged2e,booktabs,caption}
\usepackage{mwe}
\usepackage{makecell}
\usepackage{flushend}

\newcommand{\doublewidetilde}[1]{{%
		\mathpalette\double@widetilde{#1}%
	}}
	
	\makeatletter
	
	\def\BState{\State\hskip-\ALG@thistlm}
	\makeatother

	\newtheorem{theorem}{Theorem}
	\newtheorem{lemma}{Lemma}

	\newcolumntype{C}[1]{>{\Centering}m{#1}}
	
\begin{document}
		\title{Distributed Power Control in Downlink Cellular Massive MIMO Systems}

		\author{\IEEEauthorblockN{Trinh Van Chien\IEEEauthorrefmark{1}, Emil Bj\"{o}rnson\IEEEauthorrefmark{1}, Erik G. Larsson\IEEEauthorrefmark{1}, and Tuan Anh Le\IEEEauthorrefmark{2}}
			\IEEEauthorblockA{\IEEEauthorrefmark{1}Department of Electrical
				Engineering (ISY), Link\"{o}ping University, SE-581 83 Link\"{o}ping, Sweden}
			\IEEEauthorblockA{\IEEEauthorrefmark{2}Department of Design Engineering \& Maths, Middlesex University London, United Kingdom}
			\IEEEauthorblockA{\IEEEauthorrefmark{0}Email:	\{trinh.van.chien, emil.bjornson, erik.g.larsson\}@liu.se, t.le@mdx.ac.uk}
	
			\thanks{This paper was supported by the European Union's Horizon 2020 research and innovation programme under grant agreement No 641985 (5Gwireless). It was also supported by ELLIIT and CENIIT.}}
		
		\maketitle
		
		\begin{abstract}
			This paper compares centralized and distributed methods to solve the power minimization problem with quality-of-service (QoS) constraints in the downlink (DL) of multi-cell Massive multiple-input multiple-output (MIMO) systems. In particular, we study the computational complexity, number of parameters that need to be exchanged between base stations (BSs), and the convergence of iterative implementations. Although a distributed implementation based on dual decomposition (which only requires statistical channel knowledge at each BS) typically converges to the global optimum after a few iterations, many parameters need to be exchanged to reach convergence.
			\end{abstract}

		%\begin{IEEEkeywords}
		%	Massive MIMO, Distributed implementation, Linear program, Second-order cone program.
		%\end{IEEEkeywords}% no keywords
		
		\section{Introduction}
		Massive MIMO is considered a key technology for $5$G networks due to its great improvements in both spectral efficiency (SE) and energy efficiency (EE) over legacy networks \cite{Marzetta2016a}. Moreover, resource allocation problems in Massive MIMO are reported to have much lower complexity than that in small-scale systems owning to the fact that the ergodic SE expressions only depend on the large-scale fading coefficients, thanks to the so-called channel hardening property \cite{Chien2016b}. However, so far, most optimization problems in the Massive MIMO literature have been formulated and solved in a centralized fashion \cite{Chien2016b,Ngo2015a}. This requires the network to gather full statistical channel state information (CSI) from all base stations (BSs) at one location in order to centrally allocate resources in every cell and then inform each BS of the decisions. This raises practical questions, especially for dense networks with many BSs and users, about backhaul signaling, scalability, and delays \cite{Bjornson2013d}. A classic approach to deal with such issues is distributed optimization \cite{Palomar2006a}, which transforms the centralized problem implementations where every BS simultaneously optimizes its local resources based on local information and only parameters that describe the inter-cell interference are exchanged between BSs to iteratively find the globally optimal solution.
		
		A few recent works have applied distributed implementation concepts to Massive MIMO, e.g., \cite{ Adhikary2017, zappone2015,Asgharimoghaddam2014}. For the uplink (UL) transmission, a distributed max-min fairness problem for a two-decoding-layers Massive MIMO system is studied in \cite{Adhikary2017} based on the existence of an effective interference function obeying the rigid conditions \cite{Yates1995a, le2013}. By formulating a non-cooperative game, \cite{zappone2015} proposed a distributed EE problem which has a Nash equilibrium. For the downlink (DL) transmission, by utilizing the UL-DL duality, \cite{Asgharimoghaddam2014} proposed a distributed framework for a total transmit power minimization problem with quality-of-service (QoS) constraints that is also seeking the optimal precoding vectors, which are functions of the small-scale fading realizations. Even though the optimal precoding vectors provides certain gains over heuristic precoding such as zero-forcing (ZF), the algorithm in \cite{Asgharimoghaddam2014} suffers from the fact that the small-scale fading realizations vary rapidly over both time and frequency. Furthermore, the previous algorithms assumed that the BSs have full access to the channel statistics of the neighboring cells \cite{Adhikary2017} or all the other cells \cite{zappone2015,Asgharimoghaddam2014}, which may require heavy backhaul signaling in the backhaul since users move, new users arrive, and current users disconnect. To the best of our knowledge, no previous work has explicitly investigated what information should be exchanged between BSs in Massive MIMO to solve power minimization problems.
		
		In this paper, we consider the DL transmission of Massive MIMO systems with maximum ratio (MR) or ZF precoding. Each user has a QoS requirement, in terms of an achievable SE, and we study the total transmit power minimization problem. We compare a centralized solution algorithm with two distributed algorithms to answer the following fundamental questions: i) Can a distributed implementation achieve the optimal solution of the centralized problem? ii) Which subset of parameters should be shared between the BSs? iii) How can we limit complexity and backhaul signaling requirements for distributed implementation?

		\textit{Notations}: Upper/lower bold letters are used for matrices/vectors. $(\cdot)^H$ denotes the Hermitian transpose. $\mathbb{E}\{ \cdot\}$ denotes the expectation of a random variable, while $\mathcal{CN}(\cdot, \cdot)$ stands for the circularly symmetric Gaussian distribution and $\| \cdot\|$ is the Euclidean norm.
		%\vspace*{-0.3cm}
		\section{System Model} \label{Section: System Model}
		We consider a Massive MIMO system comprising $L$ cells, each having a BS equipped with $M$ antennas and serving $K$ single-antenna users. The system operates according to a time division duplex (TDD) protocol. The time-frequency resources are divided into coherence intervals of $\tau_c$ symbols where the channels are assumed to be static and frequency flat. The channel between user $k$ in cell $i$ and BS $l$ is assumed to be uncorrelated Rayleigh fading,
		\begin{equation}
		\mathbf{h}_{i,k}^l \sim \mathcal{CN}(\mathbf{0}, \beta_{i,k}^l \mathbf{I}_M),
		\end{equation}
		where $\beta_{i,k}^l$ is the large-scale fading coefficient.
		
		During the UL channel estimation phase, each coherence interval dedicates $\tau_p$ symbols for pilot transmission. We assume that the same set of $\tau_p = K$ orthonormal pilot signals $\{ \pmb{\psi}_1,\ldots,\pmb{\psi}_K \}$, with $\|\pmb{\psi}_k\|^2 = \tau_p, \forall k$, is reused in each cell wherein user~$k$ in cell $l$ allocates the power $\hat{p}_{l,k} > 0$ to its pilot signal. The received training signal at BS~$l$ is
		\begin{equation} \label{eq:PilotTraining1}
		\mathbf{Y}_l = \sum_{i=1}^L \sum_{t=1}^K \sqrt{\hat{p}_{i,t}} \mathbf{h}_{i,t}^l \pmb{\psi}_{t}^H + \mathbf{N}_l,
		\end{equation}
		where $\mathbf{N}_l$ is additive noise with each element independently distributed as $\mathcal{CN}(0, \sigma_{\mathrm{UL}}^2)$, where $\sigma_{\mathrm{UL}}^2$ is the variance. The channel between user~$k$ in cell~$l$ and BS~$l$ is estimated by multiplying the received signal in \eqref{eq:PilotTraining1} with $\pmb{\psi}_k$ as
		\begin{equation}
		\mathbf{y}_{l,k} = \mathbf{Y}_l \pmb{\psi}_k = \sum_{i=1}^L \sum_{t=1}^K \sqrt{\hat{p}_{i,t}} \mathbf{h}_{i,t}^l \pmb{\psi}_{t}^H \pmb{\psi}_{k} + \mathbf{N}_l  \pmb{\psi}_{k}.
		\end{equation}
	    Using minimum mean square error (MMSE) estimation \cite{Kay1993a, Chien2018a}, the channel estimate is distributed as
		\begin{equation}
		 \hat{\mathbf{h}}_{i,k}^l \sim \mathcal{CN}(\mathbf{0}, \gamma_{i,k}^l \mathbf{I}_M),
		\end{equation}
		where the variance $\gamma_{i,k}^l$ is
		\begin{equation}
		\gamma_{i,k}^l =  \frac{\hat{p}_{i,k} \tau_p (\beta_{i,k}^l)^2}{\sum\limits_{i' =1}^L \hat{p}_{i',k} \tau_p \beta_{i',k}^l + \sigma_{\mathrm{UL}}^2}.
		\end{equation}
		In this paper, the channel estimates are used to construct linear precoding vectors for the DL data transmission.
		%\vspace{-0.25cm}
		\subsection{Downlink Data Transmission}
		In the DL data transmission, BS $l$ transmits a Gaussian signal $s_{l,k}$ to its user $k$ with $\mathbb{E} \{ |s_{l,k} |^2 \} =1$. The received baseband signal at user $k$ in cell $l$ is
		\begin{equation}
		y_{l,k} = \sum_{i=1}^L \sum_{t=1}^K \sqrt{\rho_{i,t}} (\mathbf{h}_{l,k}^i)^H \mathbf{w}_{i,t} s_{i,t} + n_{l,k},
		\end{equation}
		where the additive noise is $n_{l,k} \sim \mathcal{CN}(0,\sigma_{\mathrm{DL}}^2)$,  $\rho_{i,t}$ is the power allocated to user $t$ in cell $i$ for transmission of the data symbol $s_{i,t}$ and $\mathbf{w}_{i,t}$ is the corresponding normalized linear precoding vector:
		\begin{equation} \label{eq: Linear-Precoding-Vector}
		\mathbf{w}_{l,k}  = \begin{cases}  \frac{ \hat{\mathbf{h}}_{l,k}^l }{ \sqrt{ \mathbb{E} \{  \| \hat{ \mathbf{h}}_{l,k}^l  \|^{2} \}}}, & \mbox{for MR,} \\
		\frac{ \widehat{\mathbf{H}}_{l}^l \mathbf{r}_{l,k} }{\sqrt{ \mathbb{E} \{ \|\widehat{\mathbf{H}}_{l}^l \mathbf{r}_{l,k}\|^{2} \} }}, & \mbox{for ZF,} \end{cases}
		\end{equation}
		where $\widehat{\mathbf{H}}_{l}^l = [\hat{ \mathbf{h}}_{l,1}^l, \ldots, \hat{ \mathbf{h}}_{l,K}^l ] \in \mathbb{C}^{M \times K}$ is the estimated channel matrix of the $K$ users in cell $l$, $\mathbf{r}_{l,k}$ is the $k$th column of $(\widehat{\mathbf{H}}_{l}^{l,H} \widehat{\mathbf{H}}_{l}^{l})^{-1}$.
		 Using standard techniques, a closed-form lower bound on the DL ergodic capacity with MR or ZF precoding is obtained.
		\begin{lemma}\cite[Corollary~$3$]{Chien2017a} \label{Lemma1}
			In the DL, the closed-form expression for the ergodic SE of user $k$ in cell $l$ is
			\begin{equation} \label{eq:DLRate}
			R_{l,k} = \left(1 - \frac{\tau_p}{\tau_c} \right) \log_2 \left( 1 + \mathrm{SINR}_{l,k}\right),
			\end{equation}
			where the effective  signal-to-interference-plus-noise ratio (SINR), denoted by $\mathrm{SINR}_{l,k}$, is
			\begin{equation}
			\mathrm{SINR}_{l,k} = \frac{G \rho_{l,k} \gamma_{l,k}^l }{G \sum\limits_{\substack{i = 1 \\ i \neq l}}^L \rho_{i,k} \gamma_{l,k}^i  + \sum\limits_{i=1}^L \sum\limits_{t=1}^K \rho_{i,t} z_{l,k}^i + \sigma_{\mathrm{DL}}^2}.
			\end{equation}
			The parameters $G$ and $z_{l,k}^i$ are specified by the precoding scheme. MR precoding gives $G=M$ and $z_{l,k}^i = \beta_{l,k}^i$, while $ZF$ precoding gives $G=M-K$ and $z_{l,k}^i = \beta_{l,k}^i -\gamma_{l,k}^i$.
		\end{lemma}
		 As the closed-form expression of the ergodic SE in \eqref{eq:DLRate} is independent of the small-scale fading, it can be used to solve resource allocation problems whose solutions are applicable over a long time period \cite{Adhikary2017, zappone2015,Asgharimoghaddam2014}. In this paper, we introduce a distributed implementation for the total transmit power minimization problem.
		 \begin{figure}[t]
		 	\centering
		 	\includegraphics[trim=5.5cm 10.3cm 10.5cm 12.3cm, clip=true, width=3.5in]{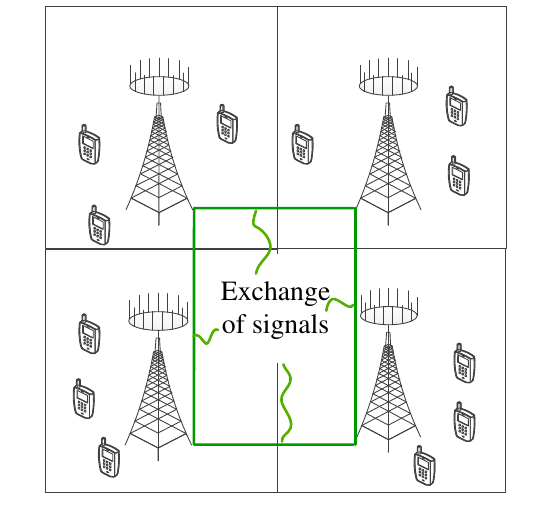} %\vspace*{-0.05cm}
		 	\caption{A multi-cell Massive MIMO system where the problem~\eqref{eq:Opt1} is optimized in distributed manner.}
		 	\label{Fig1SystemModel}
		 	%\vspace*{-0.25cm}
		 \end{figure}
		\subsection{Total Transmit Power Minimization Problem}
		The total transmit power at BS $l$ is $ \sum_{k=1}^K \rho_{l,k}, \forall l$. Suppose user $k$ in cell $l$  has the QoS requirement $R_{l,k} \geq \xi_{l,k}$, where $R_{l,k}$ is given by \eqref{eq:DLRate}.
		The total transmit power minimization problem of the $L$ BSs subject to these QoS requirements is formulated as
				\begin{equation} \label{eq:General-Form-Optimization}
		\begin{aligned}
		& \underset{ \{\rho_{l,k} \geq 0 \} }{\textrm{minimize}}
		& & \sum_{ l=1 }^{L} \sum_{k=1}^K \rho_{l,k}\\
		& \textrm{subject to}
		& &  R_{l,k} \geq \xi_{l,k},\; \forall l,k ,\\
		& && \sum_{t=1}^K \rho_{l,t} \leq P_{\mathrm{max}, l }, \; \forall l,\\
		\end{aligned}
		\end{equation}
		where $P_{\max,l}$ is the maximum transmit power at BS $l$. Converting from the SE requirements to the corresponding SINR values, i.e., set $\hat{\xi}_{l,k} = 2^{\frac{\tau_c \xi_{l,k}}{\tau_c  - \tau_p}} -1 $, and then problem \eqref{eq:General-Form-Optimization} is reformulated as
		\begin{equation} \label{eq:Opt1}
		\begin{aligned}
		& \underset{ \{ \rho_{l,k} \geq 0 \}}{ \mathrm{minimize} }  && \quad \sum_{l=1}^L \sum_{k=1}^K \rho_{l,k} \\
		&  \text{subject to} && \quad \textrm{SINR}_{l,k}\geq \hat{\xi}_{l,k}, \forall l,k, \\
		&&&  \quad \sum_{t=1}^K \rho_{l,t} \leq P_{\max, l}, \forall l.
		\end{aligned}
		\end{equation}
		Since problem \eqref{eq:Opt1} is a linear program \cite{Marzetta2016a}, its optimal solution can be obtained in polynomial time by an interior-point algorithm or a simplex method, e.g., using CVX \cite{cvx2015}. A centralized implementation requires the $KL^2$ large-scale fading coefficients of all channels (i.e., $\beta_{i,k}^l \forall i,k,l$), the $KL^2$ variances of all the channel estimates (i.e., $\gamma_{i,k}^l, \forall i,k,l$), and the $KL$ QoS requirements of the users to be gathered at one location in the network, which can then solve \eqref{eq:Opt1} to obtain the optimal power control (i.e., $\rho_{l,k}, \forall l,k$). The optimal power solution then needs to be fed back to the BSs by sending $KL$ additional parameters over the backhaul.
		
		\subsection{Basic Form of Distributed Implementation}
		
		In a distributed implementation of problem~\eqref{eq:Opt1}, each BS performs data power allocation for its own users and only iteratively exchange signals with the other BSs as shown in Fig.~\ref{Fig1SystemModel}. A basic form of \eqref{eq:Opt1} is that every BS acquires the $2K(L-1)^2 + K(L-1)$ parameters, which the centralized implementation requires, and then solves \eqref{eq:Opt1} locally. This removes the need for sending the solution over the backhaul.

		\section{Distributed Implementation by Dual Decomposition}
		In this section, a distributed implementation of \eqref{eq:Opt1} is studied based on dual decomposition. Its convergence and the computational complexity are also investigated.
		\subsection{Assumptions for Distributed Implementation}
		
		\begin{figure}[t]
			\centering
			\includegraphics[trim=3.5cm 20.6cm 9cm 4.2cm, clip=true, width=3.5in]{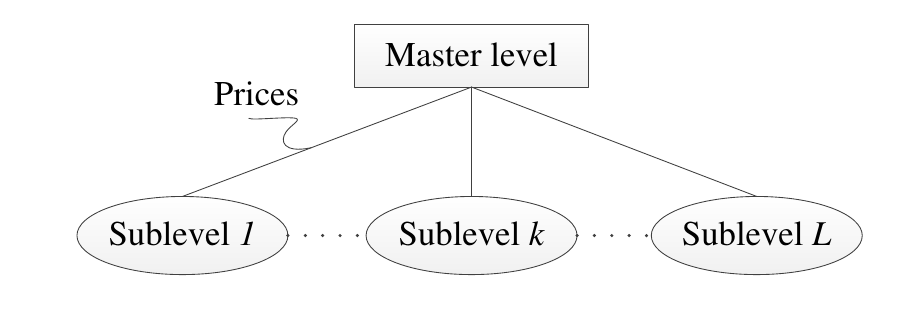} %\vspace*{-0.05cm}
			\caption{Distributed implementation of problem~\eqref{eq:Opt1} based on the dual decomposition.}
			\label{Fig1Dic}
			%\vspace*{-0.25cm}
		\end{figure}
		The proposed distributed implementation for problem \eqref{eq:Opt1} is based on two levels of optimization: a master level and the $L$ sublevels are schematically illustrated in Fig.~\ref{Fig1Dic}. At Sublevel $l$, BS $l$ will locally optimize the transmit powers to its $K$ users utilizing only partial information. The following assumptions are made to allocate power for the DL transmission: 
		\begin{enumerate}
		
		\item BS~$l$ possesses statistical information including the large-scale fading coefficients $\beta_{l,k}^l$ and the channel estimate variance $\gamma_{i,k}^l$ of the $K$ local users. 
		%\vspace*{0.15cm}
		\item BS~$l$ has the statistical information of the channels to the $K(L-1)$ interfering users, i.e., $\beta_{i,k}^l, \gamma_{i,k}^l, \forall k, i \neq l$. Those can also be measured locally.
		%\vspace*{0.15cm}
		\item BS $l$ jointly optimizes the DL powers allocated to its $K$ local users based on the above prior information. 
		%\vspace*{0.15cm}
		\item The inter-cell interference is considered as consistency parameters in a dual-decomposition approach. These are the only parameters needed to exchange between BSs and will be defined hereafter.
		
		\end{enumerate}
		
		\subsection{Details of the Distributed Implementation} \label{SubSection:DistributedTotalTransmitPower}
		Since the cost function in the optimization problem \eqref{eq:Opt1} is not strictly convex, a standard dual decomposition implementation is not guaranteed to converge \cite{Boyd2008b}.\footnote{A function $f$ is strictly convex if for two variables $x_1 \neq x_2$ in the feasible domain of $f$ and a scalar $\alpha \in (0,1)$ that gives $\alpha x_1 + (1- \alpha)x_2$ also in the feasible domain, then $f(\alpha x_1 + (1- \alpha)x_2) < \alpha f(x_1) +(1- \alpha)f(x_2).$ } Thus, in order to ensure the convergence of the distributed implementation, we will convert \eqref{eq:Opt1} into a convex problem which involves a strictly convex cost function by introducing a new variable $\tilde{\rho}_{l,k} \geq 0 $ such as $\rho_{l,k} = \tilde{\rho}_{l,k}^2, \forall l,k$. The effective SINR value of user $k$ in cell $l$ is reformulated as
		\begin{equation} \label{eq:ReformulatedSINR}
		\begin{split}
		 &\mathrm{SINR}_{l,k} = \frac{G \tilde{\rho}_{l,k}^2 \gamma_{l,k}^l }{ G \sum\limits_{\substack{i=1 \\ i \neq l}}^L  \tilde{\rho}_{i,k}^2 \gamma_{l,k}^i + \sum\limits_{i=1}^L \sum\limits_{t=1}^K \tilde{\rho}_{i,t}^2 z_{l,k}^i  + \sigma_{\mathrm{DL}}^2} \\
		  &= \frac{G \tilde{\rho}_{l,k}^2 \gamma_{l,k}^l }{ \sum\limits_{t=1}^K \tilde{\rho}_{l,t}^2 z_{l,k}^l + \sum\limits_{\substack{i=1 \\ i \neq l}}^L\left(G \tilde{\rho}_{i,k}^2 \gamma_{l,k}^i + \sum\limits_{t=1}^K \tilde{\rho}_{i,t}^2 z_{l,k}^i \right) + \sigma_{\mathrm{DL}}^2}.
		\end{split}
		\end{equation}
		In the last equation of \eqref{eq:ReformulatedSINR}, the first term in the denominator represents the mutual non-coherent interference of the $K$ local users served by BS~$l$. Hence, BS~$l$ can locally evaluate this term. The second term includes the coherent interference caused by the users utilizing nonorthogonal pilot signals and the non-coherent interference from all users in the other cells. If BS~$l$ wishes to evaluate the SE of each local user, it will need to obtain this information from the other $L-1$~BSs (i.e., the prices in Fig.~\ref{Fig1Dic}) to compute the second part. In order to reduce the amount of exchanged information among BSs, we introduce  so-called consistency variables $\theta_{l,k}^{i}$ and $\tilde{\theta}_{l,k}^i$ which represent the exact and believed value of $G \tilde{\rho}_{i,k}^2 \gamma_{l,k}^i + \sum_{t=1}^K \tilde{\rho}_{i,t}^2 z_{l,k}^i$, respectively. Thus, \eqref{eq:Opt1} is reformulated as
		\begin{subequations} \label{eq:Opt2}
			\begin{alignat}{2}
			&&& \underset{ \{ \tilde{\rho}_{l,k} \geq 0, \theta_{l,k}^i \geq 0, \tilde{\theta}_{l,k}^i \geq 0 \}}{ \mathrm{minimize} }   \quad \sum_{i=1}^L \sum_{t=1}^K \tilde{\rho}_{i,t}^2 \\
			&&&  \text{subject to}  \; \frac{G \tilde{\rho}_{l,k}^2 \gamma_{l,k}^l }{ \sum\limits_{t=1}^K \tilde{\rho}_{l,t}^2 z_{l,k}^l + \sum\limits_{\substack{i=1 \\ i \neq l}}^L (\tilde{\theta}_{l,k}^i)^2 + \sigma_{\mathrm{DL}}^2}  \geq \hat{\xi}_{l,k}, \forall l,k, \\
			&&&   G \tilde{\rho}_{i,k}^2 \gamma_{l,k}^i + \sum\limits_{t=1}^K \tilde{\rho}_{i,t}^2 z_{l,k}^i \leq (\theta_{l,k}^i)^2, \forall l,i,k, i \neq l,\\
			&&& \theta_{l,k}^i \leq \tilde{\theta}_{l,k}^i,\forall l,i,k, i \neq l, \label{eq:Opt2d}\\
			&&&  \sum_{k=1}^K \tilde{\rho}_{l,k}^2 \leq P_{\max, l}, \forall l.
			\end{alignat}
		\end{subequations}
	The constraints \eqref{eq:Opt2d} are called consistency constraints. There are in total $2L(L-1)K$ consistency variables $\theta_{l,k}^{i}$ and $\tilde{\theta}_{l,k}^{i}, \forall k, i,l$, thus \eqref{eq:Opt2} involves $LK(2LK-1)$ optimization variables. We stress that \eqref{eq:Opt1} and \eqref{eq:Opt2} are equivalent since $\theta_{l,k}^{i} = \tilde{\theta}_{l,k}^{i}$ at the optimum. The dual decomposition approach is used to decompose problem \eqref{eq:Opt2} into $L$ subproblems that each can be solved locally at a BS. To that end, a partial Lagrangian function, which is related to the differences between the consistency variables, is formed as
		\begin{equation} \label{eq:LagrangianOfTotalTransmitPower}
		\begin{split}
		&\mathcal{L}\left(\{\tilde{\rho}_{l,k}, \lambda_{l,k}^i \}\right)  = \sum_{l=1}^L \sum_{k=1}^K \left(\tilde{\rho}_{l,k}^2 + \sum_{\substack{i=1 \\ i \neq l}}^L \lambda_{l,k}^i ( \theta_{l,k}^i  - \tilde{\theta}_{l,k}^i)\right),
		\end{split}
		\end{equation}
		where $\lambda_{l,k}^{i} \geq 0$ is the Lagrange multiplier associated with the constraint $\theta_{l,k}^i \leq \tilde{\theta}_{l,k}^i$. The dual function of \eqref{eq:LagrangianOfTotalTransmitPower} is computed as the superposition of the $L$ local dual functions
		\begin{equation}
		g \left( \{ \lambda_{l,k}^i \} \right) = \sum_{l=1}^L g_l ( \{ \lambda_{l,k}^i \} ),
		\end{equation}
		where the local dual function of BS $l$, denoted by $g_l ( \{ \lambda_{l,k}^i \})$, is formulated as
		\begin{equation} \label{eq:LocalDualofTotalTransmitPower}
	\begin{split}		 
	&g_l ( \{ \lambda_{l,k}^i \}) = \\
	&\underset{ \substack{\{\tilde{\rho}_{l,k} \geq 0, \theta_{i,k}^l \geq 0, \tilde{\theta}_{l,k}^i \geq 0  \}}}{\mathrm{inf}} \; \sum_{k=1}^K \tilde{\rho}_{l,k}^2 + \sum_{k=1}^K \sum_{\substack{i=1\\ i\neq l }}^{L} \left(\lambda_{i,k}^l \theta_{i,k}^l - \lambda_{l,k}^i \tilde{\theta}_{l,k}^i \right).
	\end{split}
		\end{equation}
		From \eqref{eq:LocalDualofTotalTransmitPower}, problem \eqref{eq:Opt2} can be decomposed into the $L$ subproblems and the $l$th subproblem is
		\begin{equation} \label{eq:Subproblem}
		\begin{aligned}
		&&& \underset{ \substack{\{ \tilde{\rho}_{l,k} \geq 0, \theta_{i,k}^l \geq 0, \tilde{\theta}_{l,k}^i \geq 0 \}}}{ \mathrm{minimize} }  \sum_{k=1}^K \tilde{\rho}_{l,k}^2 + \sum_{k=1}^K \sum_{\substack{i=1\\ i\neq l }}^{L}(\lambda_{i,k}^l \theta_{i,k}^l - \lambda_{l,k}^i \widetilde{\theta}_{l,k}^i ) \\
		&&&  \text{subject to} \; \frac{G \tilde{\rho}_{l,k}^2 \gamma_{l,k}^l }{ \sum\limits_{t=1}^K \tilde{\rho}_{l,t}^2 z_{l,k}^l + \sum\limits_{\substack{i=1 \\ i \neq l}}^L (\tilde{\theta}_{l,k}^i)^2 + \sigma_{\mathrm{DL}}^2}  \geq \hat{\xi}_{l,k}, \forall k, \\
		&&&   G \tilde{\rho}_{l,k}^2 \gamma_{i,k}^l + \sum\limits_{t=1}^K \tilde{\rho}_{l,t}^2 z_{i,k}^l \leq (\theta_{i,k}^l)^2, \forall i,k, i \neq l, \\
		&&& \sum_{k=1}^K \tilde{\rho}_{l,k}^2 \leq P_{\max, l}.
		\end{aligned}
		\end{equation}
	Notice that BS $l$ uses $K(2L-1)$ optimization variables to solve the $l$th subproblem, while all the $L$ subproblems can be processed in parallel by the $L$ BSs for given values of all the Lagrange multipliers $\lambda_{l,k}^{i}$. The globally optimal solution to \eqref{eq:Subproblem} is obtained due to its convexity as stated in Theorem~\ref{Theorem:SOCP}.
		\begin{theorem} \label{Theorem:SOCP}
			The optimization problem \eqref{eq:Subproblem} is equivalent to the following second-order cone (SOC) program
			\begin{equation} \label{eq:SOCP}
			\begin{aligned}
			&&& \underset{ \substack{\{ \tilde{\rho}_{l,k} \geq 0\},s_l, \{ \theta_{i,k}^l \geq 0, \tilde{\theta}_{l,k}^i \geq 0 \}}}{ \mathrm{minimize} }  s_l \\
			&&&  \text{subject to} \;   \| \hat{\mathbf{u}}_{l,k} \| \leq  \tilde{\rho}_{l,k}\sqrt{\frac{G \gamma_{l,k}^l}{\hat{\xi}_{l,k}}}, \forall k, \\
			&&& \| \mathbf{u}_l \| \leq  \frac{1}{2} \left( 1 +s_l + \sum_{k=1}^K \sum_{i=1, i\neq l }^{L}(\lambda_{l,k}^i \tilde{\theta}_{l,k}^i - \lambda_{i,k}^l \theta_{i,k}^l) \right),\\
			&&&   \| \tilde{\mathbf{u}}_{i,k}^l \| \leq \theta_{i,k}^l, \forall i,k, i \neq l, \\
			&&& \| \breve{\mathbf{u}}_l \| \leq \sqrt{P_{\max, l}},
			\end{aligned}
			\end{equation}
			where $\hat{\mathbf{u}}_{l,k} \in \mathbb{C}^{K+L}, \mathbf{u}_l \in \mathbb{C}^{K+1}, \tilde{\mathbf{u}}_{i,k}^l \in \mathbb{C}^{K+1},$ and $\breve{\mathbf{u}}_l \in \mathbb{C}^K$ are respectively defined as
			\begin{align*}
			& \hat{\mathbf{u}}_{l,k} = \\
			&\left[ \sqrt{z_{l,k}^l}(\tilde{\rho}_{l,1}, \ldots, \tilde{\rho}_{l,K}), \tilde{\theta}_{l,k}^1, \ldots,\tilde{\theta}_{l,k}^{l-1},\tilde{\theta}_{l,k}^{l+1}, \ldots, \tilde{\theta}_{l,k}^{L}, \sigma_{\mathrm{DL}}  \right]^T,\\
			&\mathbf{u}_l = \\
			& \left[\tilde{\rho}_{l,1}, \ldots, \tilde{\rho}_{l,K}, \frac{1}{2} \left( 1 - s_l + \sum_{k=1}^K \sum_{\substack{i=1,\\ i\neq l}}^{L}(\lambda_{i,k}^l \theta_{i,k}^l -\lambda_{l,k}^i \tilde{\theta}_{l,k}^i )\right) \right]^T,  \\
			&\tilde{\mathbf{u}}_{i,k}^l =  \left[ \tilde{\rho}_{l,k} \sqrt{G \gamma_{i,k}^l}, \sqrt{z_{i,k}^l}(\tilde{\rho}_{l,1}  , \ldots,  \tilde{\rho}_{l,K}) \right]^T, \\
			&\breve{\mathbf{u}}_l = [\tilde{\rho}_{l,1}, \ldots, \tilde{\rho}_{l,K} ]^T.
			\end{align*}
		\end{theorem}
		\begin{proof}
			Utilizing the epi-graph representation \cite{Boyd2004a}, the objective function of \eqref{eq:Subproblem} turns to the following constraint by introducing a new optimization variable $s_l$,
			\begin{equation} 
			\sum_{k=1}^K \tilde{\rho}_{l,k}^2 + \sum_{k=1}^K \sum_{i=1, i\neq l }^{L}(\lambda_{i,k}^l \theta_{i,k}^l - \lambda_{l,k}^i \tilde{\theta}_{l,k}^i ) \leq s_l,
			\end{equation}
			and it is equivalent to
			\begin{equation}\label{eq:ObjectConstraint} 
			\sum_{k=1}^K \tilde{\rho}_{l,k}^2  \leq s_l - \sum_{k=1}^K \sum_{i=1, i\neq l }^{L}(\lambda_{i,k}^l \theta_{i,k}^l - \lambda_{l,k}^i \tilde{\theta}_{l,k}^i ).
			\end{equation}
			By applying the identity $x - y = \frac{1}{4}(1+x-y)^2 - \frac{1}{4}(1 -x +y)^2$, i.e., with $x=s_l$ and $y=\sum_{k=1}^K \sum_{i=1, i\neq l }^{L}(\lambda_{i,k}^l \theta_{i,k}^l - \lambda_{l,k}^i \tilde{\theta}_{l,k}^i )$, for the right-hand side of \eqref{eq:ObjectConstraint}, this constraint can be converted to a SOC constraint. The other constraints of \eqref{eq:Subproblem} can be easily converted to the corresponding SOC constraints as in \eqref{eq:SOCP}.
		\end{proof}
		At iteration $n+1$, after obtaining the optimal solutions to the $L$ subproblems in \eqref{eq:SOCP}, every BS will update the Lagrange multipliers at the so-called master level by considering the following master dual problem:
		\begin{equation} \label{eq:MasterDualProb}
		\begin{aligned}
		& \underset{ \{ \lambda_{l,k}^{i,(n+1)} \geq 0 \}}{ \mathrm{maximize} }  && \sum_{l=1}^L \sum_{k=1}^K \sum_{\substack{i=1 \\ i \neq l}}^L \lambda_{l,k}^{i,(n+1)} \left( \theta_{l,k}^{i,\ast,(n)}  - \tilde{\theta}_{l,k}^{i,\ast,(n)} \right),
		\end{aligned}
		\end{equation}
		where $\theta_{l,k}^{i,\ast,(n)}$ and $\tilde{\theta}_{l,k}^{i,\ast,(n)}$ are the global optimums in the $n$th iteration to the $L$ subproblems in \eqref{eq:SOCP}. The subgradient projection method \cite{Palomar2006a} can be adopted to update the Lagrangian multipliers as
		\begin{equation} \label{eq:LagrangeMultiplierUpdate}
		\lambda_{l,k}^{i,(n+1)} = \left[\lambda_{l,k}^{i,(n)} - \varsigma ^{(n)} \left( \tilde{\theta}_{l,k}^{i,\ast, (n)} - \theta_{l,k}^{i,\ast, (n)}  \right) \right]_+, \forall l,i,k,
		\end{equation}
		where $\varsigma^{(n)}$ is a positive step-size at the $n$th iteration and $[\cdot]_+$ is the projection onto the nonnegative orthant. If the master problem is solved at an arbitrary BS, it acquires the $2K(L-1)^2$ consistency parameters from the $L-1$ remaining BSs. The $2K(L-1)$ updated Lagrange multiplier $\lambda_{i,k}^{l,(n)}$ and $\lambda_{l,k}^{i,(n)}, \forall i,l,k, i\neq L$ should be sent back to BS $l$ in every iteration. In total, the number of exchanged parameters for each iteration is $4K(L-1)^2$. The proposed distributed implementation of problem \eqref{eq:Opt1} is presented in Algorithm~\ref{Algorithm:DistributedTotalTransmitPower}.
		\begin{algorithm}[t]
			\caption{Distributed implementation for problem \eqref{eq:Opt1}} \label{Algorithm:DistributedTotalTransmitPower}
			\textbf{Input}: $P_{\max,l}, \forall l$;  $\xi_{l,k}, \forall l,k$;  $\beta_{l,k}^i, \forall, l,i,k$; Initial Lagrange multipliers $\lambda_{l,k}^{i,(1)}, \forall l,i,k$;  Set up $n=1$ and initialize $\varsigma^{(1)}$.
			\begin{itemize}
				\item[1.] \emph{Iteration} $n$:
				\begin{itemize}
					\item[1.1.] BS $l$ solves problem \eqref{eq:SOCP} utilizing $\lambda_{l,k}^{i, (n)}$ from the master level.
					\item[1.2.] BS $l$ store the currently optimal powers: $\rho_{l,k}^{\ast,(n)} = \tilde{\rho}_{l,k}^{\ast,(n)} , \forall l,k$.
					\item[1.2.] BS $l$ send the optimal values $\tilde{\theta}_{l,k}^{i,\ast,(n)}$ and $\theta_{l,k}^{i,\ast,(n)}, \forall l,i,k,$ to the master level
					\item[1.3.] Update $\lambda_{l,k}^{i, (n+1)}$ by using \eqref{eq:LagrangeMultiplierUpdate} and then send back to BS $l$.
				\end{itemize}
				\item[2.] If \textit{Stopping criterion is satisfied $\rightarrow$ Stop}. Otherwise, go to Step 3.
				\item[3.] Set $n = n+1$ and update $\varsigma^{(n+1)}$, then go to Step $1$.
			\end{itemize}
			\textbf{Output}: The optimal solutions: $\rho_{l,k}^{\ast} = \rho_{l,k}^{(n)} \forall l,k.$
		\end{algorithm}
		 Furthermore, the convergence property of Algorithm~\ref{Algorithm:DistributedTotalTransmitPower} is established in the following theorem.
		\begin{theorem} \label{Theorem:ConvergenceOfTotalTransmitPower}
		  If the step-size $\varsigma^{(n)}$ is small enough, the distributed implementation in Algorithm~\ref{Algorithm:DistributedTotalTransmitPower} yields the optimal solution to \eqref{eq:Opt1}.
		\end{theorem}
		\begin{proof}
		 The constraint functions of problem \eqref{eq:SOCP} are strictly convex, so the solution to the master dual problem \eqref{eq:MasterDualProb} converges to the global optium with $\theta_{l,k}^i -\tilde{\theta}_{l,k}^i \rightarrow 0, \forall l,i,k,$ \cite{Pennanen2011a, Bjornson2013d}. Consequently, the centralized and distributed implementation have the same optimal solution.
		\end{proof}
		\subsection{Complexity Analysis of \eqref{eq:SOCP}}
		Since problem \eqref{eq:SOCP} contains SOC constraints, a standard interior-point method (IPM) \cite{Boyd_convex,Ben-tal_13} can be used to find its optimal solution. The worst-case runtime of the IPM is computed as follows.
		
		{\it Definition 1:} For a given tolerance $\epsilon >0$, the set of $\tilde{\rho}_{l,k}^{\epsilon},s_l^{\epsilon},  \theta_{i,k}^{l,\epsilon}, \tilde{\theta}_{l,k}^{i,\epsilon}$ is called an $\epsilon$-solution to problem \eqref{eq:SOCP} if
		\begin{equation}
		 s_l^{\epsilon} \leq s_{l}^{\ast} + \epsilon,
		\end{equation}
		 where $s_{l}^{\ast}$ is the globally optimal solution to the optimization problem \eqref{eq:SOCP}.
		%\begin{equation}
		%s_l^{\epsilon} \leq \underset{ \substack{\{ \tilde{\rho}_{l,k} \geq 0\},s_l, \{ \Theta_{i,k}^l \geq 0, \widetilde{\Theta}_{l,k}^i \geq 0 \}}}{ \mathrm{minimize} }  s_l+\epsilon.
		%\end{equation}
		
		The number of decision variables of problem \eqref{eq:SOCP} is on the order of $m=\mathcal{O}\left( K(2L-1)+1\right)$ where $\mathcal{O}(\cdot)$ represents the big-O notation, thus we obtain Lemma~\ref{comlemma}.
		\begin{lemma}
			\label{comlemma}
			The computational complexity to obtain $\epsilon$-solution to problem \eqref{eq:SOCP} is
			\begin{eqnarray}
			\label{complexAI}
			\delta \left(LK^3+6LK^2+KL^2+6LK+ 3K + 5 + m^2\right) m,
			\end{eqnarray}
			where $\delta = \ln(\epsilon^{-1})\sqrt{2LK+4}$.
		\end{lemma}
		\begin{proof}
			First, problem \eqref{eq:SOCP} has  $K$ SOC constraints of dimension $(K+L+1)$, $(L-1)K+1$ SOC constraints of dimension $(K+2)$, and $1$ SOC constraints of dimension $(K+1)$. Based on these observations, one can follow the same steps as in \cite[Section V-A]{Kun-Yu2014} to arrive at \eqref{complexAI}. Note that the term $\delta$ in \eqref{complexAI} is the order of the number of iteration required to reach $\epsilon$-solution to problem \eqref{eq:SOCP} while the remaining terms represent the per-iteration computation costs \cite{Kun-Yu2014, le2017robust}.
		\end{proof}
		
		\subsection{Numerical Results} \label{Sec:NumericalResults}
		
		Next, we study the convergence of the dual decomposition approach.
		We consider a wrapped-around cellular network with $L = 4, M=500, \tau_c =200, P_{\max,l}=40$ W, $\hat{p}_{l,k}= 200$ mW, and $\xi_{l,k} = 0.5$ b/s/Hz, $\forall l,k$. The $KL$ users are randomly distributed over the coverage area with uniform distribution. The distance of user $k$ in cell $l$ to BS $l$ is denoted as $d_{l,k}^l$ km and $d_{l,k}^l \geq 0.035$ km. The system bandwidth and noise figure are $20$ MHz
		and $-96$ dBm, respectively. The large-scale fading coefficient is modeled as
		\begin{equation}
		\beta_{l,k}^l [dB] = -148.1 - 37.6 \log_{10} (d_{l,k}^l) + z_{l,k}^l,
		\end{equation}
		where the shadow fading $z_{l,k}^l$ follows a log-normal Gaussian distribution with standard variation $7$ dB. The Lagrange multipliers are initially selected as $\lambda_{l,k}^{i,(1)}=0, \forall l,i,k$, while the step size in \eqref{eq:LagrangeMultiplierUpdate} is defined as $\varsigma^{(n)} =0.01$. ZF precoding is used for simulations.  Monte-Carlo results are obtained over $1000$ different realizations of user locations.
		
		Fig.~\ref{Fig2Dic} shows the probability of the number of iterations which Algorithm~\ref{Algorithm:DistributedTotalTransmitPower} needs to obtain the $95$\% of the global optimum for (a) $K= 1$ and (b) $K = 10$. Algorithm~\ref{Algorithm:DistributedTotalTransmitPower} converges to the desired objective value after a few iterations. For example, with $K=10$, $12.4\%$ of the realizations of user locations only need one iterations to attain $95$\% of the global optimum. If we set $400$ iterations as a maximum level to terminate the algorithm for the realizations of users which are in trouble to obtain $95$\% of the global optimum, there are $1.4\%$ of the user realizations facing with this issue. Meanwhile, by comparing Fig.~\ref{Fig2Dic}a and Fig.~\ref{Fig2Dic}b, we observe that Algorithm~\ref{Algorithm:DistributedTotalTransmitPower} requires more iterations to converge to the global optimum when increasing the number of users in the network.
		
		\begin{figure}[t]
			\centering
			\begin{minipage}{0.5\columnwidth}
				\centering
				\includegraphics[trim=6.6cm 9.2cm 6.5cm 9.6cm, clip=true, width=1.85in]{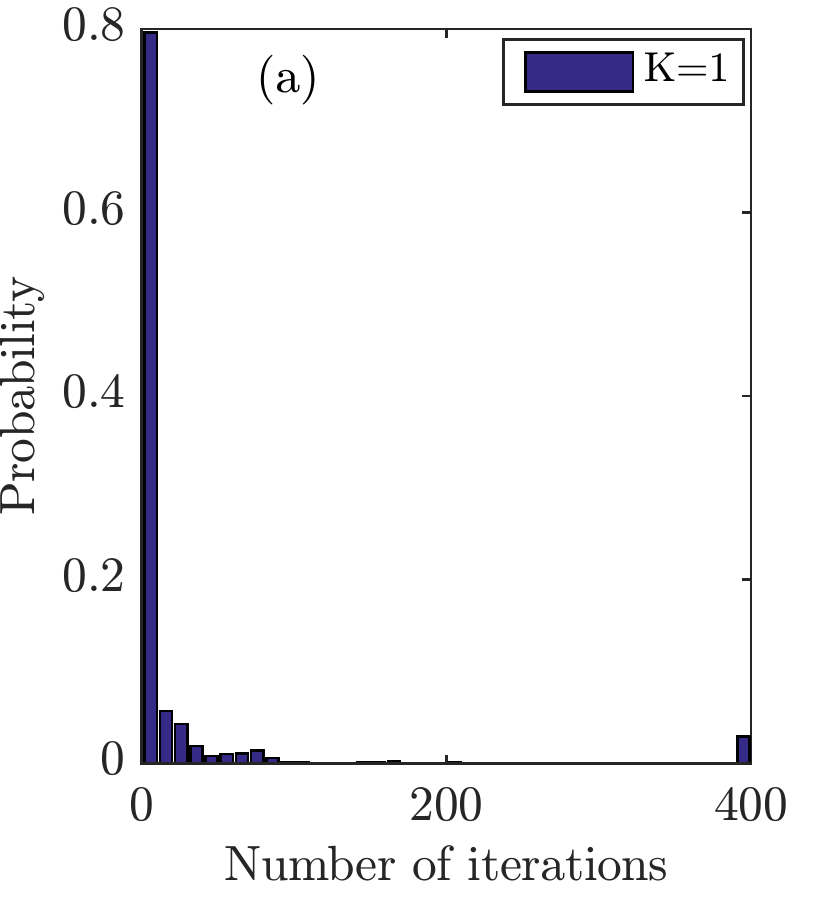} %\vspace*{-0.0cm}
			\end{minipage}\hfill
			\begin{minipage}{0.5\columnwidth}
				\centering
				\includegraphics[trim=6.6cm 9.2cm 6.5cm 9.6cm, clip=true, width=1.85in]{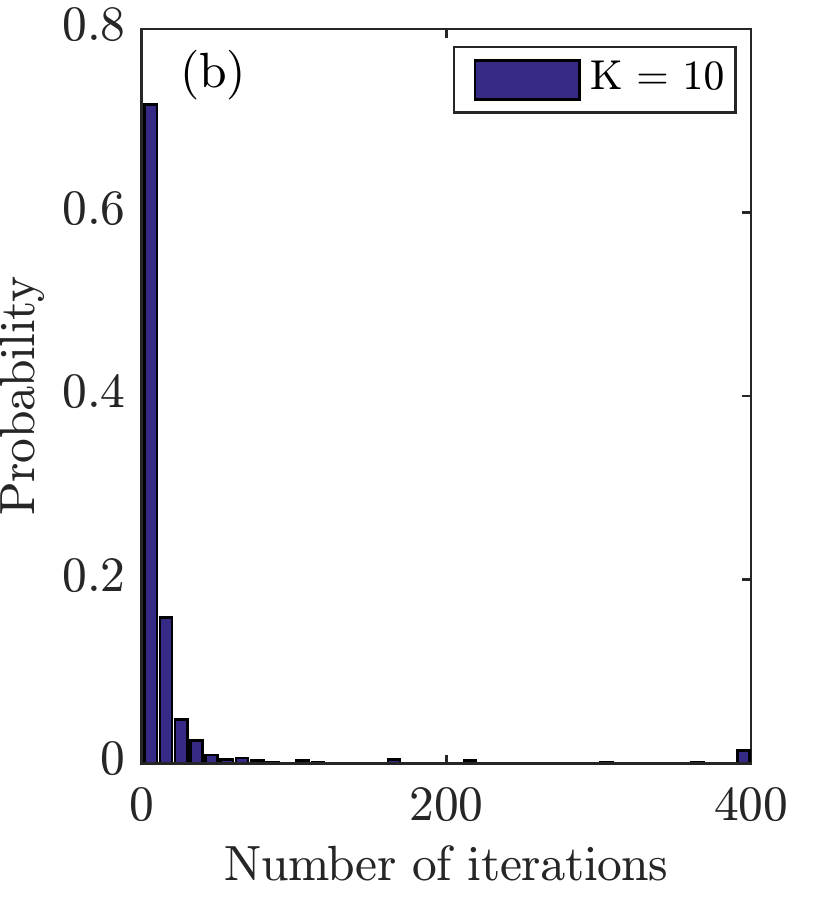} %\vspace*{-0.0cm}
			\end{minipage}
			\caption{The distribution of the number of iterations required for Algorithm~\ref{Algorithm:DistributedTotalTransmitPower} to obtain $95$\% of the global optimum: $(a) K= 1$; $(b) K = 10$.}
			\label{Fig2Dic}
			%\vspace*{-0.25cm}
		\end{figure}
		\begin{figure}[t]
			\centering
			\includegraphics[trim=4cm 8.2cm 4.5cm 9.2cm, clip=true, width=3.5in]{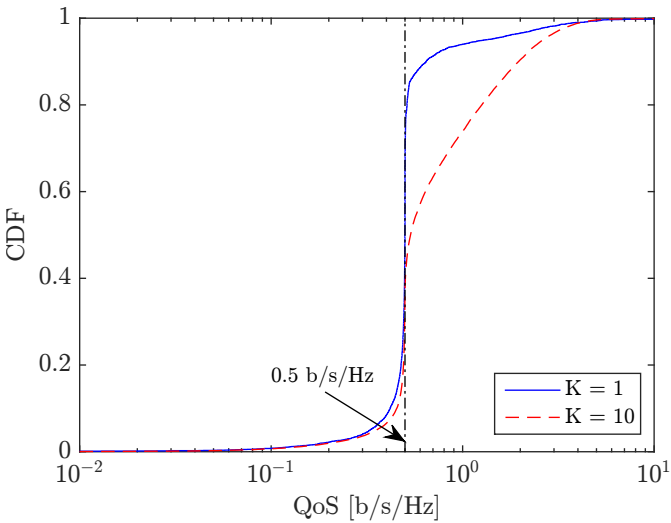} %\vspace*{-0.35cm}
			\caption{Actual QoS [b/s/Hz] of each user at the $95$\% of the global optimum.}
			\label{FigDic}
			%\vspace*{-0.25cm}
		\end{figure}
		Fig.~\ref{FigDic} gives the cumulative distribution probability (CDF) of the actual QoS per user when the algorithm has reached $95$\% of the global optimum. The results confirm that the proposed distributed implementation satisfies the QoS requirements of most of the users, while some get somewhat lower QoS and some get higher. However, the performance gap between the proposed algorithm and its centralized counterpart is wider as the number of users increases, but mainly in the sense that users get higher QoS than required.

		\section{Comparison of Centralized and Distributed Implementation Approaches}
		\begin{table}[t]
			\centering
			\caption{The total number of optimization variables and exchanged variables for the considered approaches} 
			\begin{tabular}{|c|c|c|c|}
			\hline
			Method	&\thead{ Execution \\ Place} & \thead{Total Number of \\ Opt. variables} & \thead{Total Number of \\ Ex. Parameters}  \\
			\hline
			\thead{Centralized\\
			Imp.} & \thead{Core \\ Network}& $KL$ & $2KL^2 + 2KL$\\
			\hline
			\thead{Basic form  of\\
				Distributed Imp.} & \thead{Base \\ Stations} & $KL^2$ & \thead{$2K(L-1)^2 L$ \\ $+ K(L-1)L$}\\
			\hline
		\thead{Distributed Imp. \\ 
		by Dual Decomp.}& \thead{Base \\Stations} & \thead{$2KL^2$\\ $- KL + L$} & \thead{$4K(L-1)^2N$ \\ $+2K(L-1)L$}\\
	\hline
			\end{tabular}
			\label{Table1}
		\end{table}
		Table~\ref{Table1} summarizes the total number of the optimization variables and the exchanged parameters for the centralized and the two distributed implementations. In this table, $N$ represents the number of iterations used in the dual decomposition algorithm to reach convergence. Even though the distributed approach utilizing dual decomposition divides \eqref{eq:Opt1} into the $L$ subproblems where every BS locally optimizes the power control coefficients, it also introduces new optimization variables leading to a large number of optimization variables, both in total and when comparing each subproblem with the centralized problem. Note that the global problem is a linear program, while the subproblems are SOC programs, which are more complex to solve even when the number of optimization variables are equal.
		
The total number of exchanged parameters with the dual decomposition approach may be less than in the basic form of distributed implementation, because $N$ is usually small as shown in Section~\ref{Sec:NumericalResults}. However, in situations where the number of iterations needed for the dual decomposition approach increases, the number of exchanged parameters may surpass that of the two other implementations.

Unlike in small-scale networks, where the small-scale fading coefficient from every antenna to every user needs to be exchanged to solve the centralized problem, in Massive MIMO only a few statistical parameters need to be exchanged per user. Hence, in Massive MIMO, the distributed implementation by dual decomposition does not bring any substantial reductions in backhaul signaling compared to the centralized counterpart; this is basically a consequence of the channel hardening property \cite{Marzetta2016a}.

		\section{Conclusion}
		
This paper compared distributed and centralized implementations of the total transmit power minimization problem with QoS requirements in multi-cell Massive MIMO systems. In the distributed implementation based on dual decomposition, every BS can independently and simultaneously minimize its transmit powers, while controlling inter-cell interference based on to the Lagrange multipliers $\{\lambda_{l,k}^i\}$ provided by a central entity, while a subset of nominal parameters representing the strength of mutual interference $\{ \theta_{l,k}^i, \tilde{\theta}_{l,k}^i\}$ should be sent to the central entity to update these Lagrange multipliers. This iterative method converges to the global optimum after a few iterations (in the most realizations of the user locations) and its convergence is theoretically guaranteed. 
In small-scale MIMO networks, this approach can substantially reduce the computational complexity and backhaul signaling. However, in Massive MIMO systems, this distributed implementation does not bring any significant reductions in the amount of exchanged information, since each channel is only described by a few parameters. Therefore, for resource allocation problems that only involve large-scale fading coefficients, a centralized implementation is preferable in terms of both backhaul signaling and complexity.
In practice, a combination of the two approaches might be preferable, for example, where the resource allocation is not optimized on every BS or on a single central entity, but at multiple central entities that are responsible for larger coverage areas.

%\begin{multicols}{0}
%\raggedend
\bibliographystyle{IEEEtran}
\bibliography{IEEEabrv,refs}
%\end{multicols}
	\end{document}